\documentclass[letterpaper]{article} 
\usepackage{aaai2026}  
\usepackage{times}  
\usepackage{helvet}  
\usepackage{courier}  
\usepackage[hyphens]{url}  
\usepackage{graphicx} 
\usepackage{natbib}  
\usepackage{caption} 
\usepackage{algorithm}
\usepackage{algorithmic}

\usepackage{newfloat}
\usepackage{listings}
\DeclareCaptionStyle{ruled}{labelfont=normalfont,labelsep=colon,strut=off} 
\floatstyle{ruled}
\newfloat{listing}{tb}{lst}{}
\floatname{listing}{Listing}
\usepackage[utf8]{inputenc} 
\usepackage{url}            
\usepackage{booktabs}       
\usepackage{amsfonts, amsmath, amsthm}       
\usepackage{nicefrac}       
\usepackage{microtype}      
\usepackage{xcolor}         
\usepackage{graphicx}
\usepackage{changepage}
\usepackage{comment}
\newtheorem{proposition}{Proposition}
\newtheorem{theorem}{Theorem}
\usepackage{listings} 

\usepackage{makecell}
\usepackage{multirow}
\newcommand{\WS}{\mathrm{WinklerScore}}

\nocopyright

\newtheorem{theory}{Theory}

\title{LLMs are Overconfident: Evaluating Confidence Interval Calibration with FermiEval}

\author{
Elliot L.\ Epstein\textsuperscript{\rm 1},
John Winnicki\textsuperscript{\rm 1},
Thanawat Sornwanee\textsuperscript{\rm 1},
Rajat Dwaraknath\textsuperscript{\rm 1}
}
\affiliations{
\textsuperscript{\rm 1}Stanford University, Stanford, CA 94305, USA\\
\{epsteine, winnicki, tsornwanee, rajatvd\}@stanford.edu
}

\begin{document}

\maketitle

\begin{abstract}
    Large language models (LLMs) excel at numerical estimation but struggle to correctly quantify uncertainty. We study how well LLMs construct confidence intervals around their own answers and find that they are systematically overconfident. To evaluate this behavior, we introduce FermiEval, a benchmark of Fermi-style estimation questions with a rigorous scoring rule for confidence interval coverage and sharpness. Across several modern models, nominal 99\% intervals cover the true answer only 65\% of the time on average. With a conformal prediction based approach that adjusts the intervals, we obtain accurate 99\% observed coverage, and the Winkler interval score decreases by 54\%. We also propose direct log-probability elicitation and quantile adjustment methods, which further reduce overconfidence at high confidence levels. Finally, we develop a perception-tunnel theory explaining why LLMs exhibit overconfidence: when reasoning under uncertainty, they act as if sampling from a truncated region of their inferred distribution, neglecting its tails.

\end{abstract}

\section{Introduction}

Large language models (LLMs) have shown remarkable success in several areas of mathematics. One important challenge that has been less explored is how well LLMs can construct confidence intervals around their own answer estimates. Being able to do this accurately is challenging as it necessitates the model to understand the limitations of its knowledge. Confidence intervals (CIs) provide a clear test: a nominal $95\%$ CI should achieve close to $95\%$ coverage under repeated evaluation. We study this question on a large suite of Fermi-style problems (estimation tasks with order-of-magnitude ground truths) precisely because they couple accessible prompts with unambiguous, scalar answers.

We introduce \textbf{FermiEval}, a benchmark and protocol for eliciting point estimates and target-level CIs from modern LLMs and for scoring them by both \emph{coverage} and \emph{efficiency} (informativeness for a given width). Across multiple models and confidence targets, our experiments \emph{reveal} systematic departures from nominal guarantees: observed coverage often plateaus well below target levels even as the stated confidence increases. To explain this pattern, we propose a simple mathematical account—the \emph{perception-tunnel} hypothesis—in which the model behaves as if reasoning over a truncated slice of its inferred distribution. The theory predicts the kind of stagnant coverage we measure and motivates proper scoring rules tailored to interval quality.

\paragraph{Contributions.}
(1) An evaluation protocol for CI calibration on Fermi problems; 
(2) empirical evidence that state-of-the-art LLMs \emph{do not, in general, achieve} their stated CI coverage on this task, with coverage saturating below nominal levels; 
(3) A method based on conformal prediction to adjust the confidence intervals of LLMs to reach their nominal level; (4) a formal framework showing how truncated (“tunneled’’) beliefs can yield the observed behavior and how to evaluate intervals via a proper interval score.
Taken together, these results position CI calibration as a concrete, testable facet of \emph{mathematical reasoning} in LLMs.

\section{Method}
This section presents several proposed methods to improve the calibration of LLMs. One additional method is presented in Appendix~\ref{multi-sampling}.
\subsection{Conformal Confidence Interval Adjustment}
\label{conformal_coverage}

Conformal prediction provides finite-sample coverage guarantees without assumptions on model calibration or distributional form. 
\paragraph{Base intervals.}
Assume that for each input $x$, the model outputs a base $(1-\alpha)$ prediction interval $[L(x), U(x)]$. 
This interval may be obtained in one of two ways:
(i) by sampling the model $M$ times at a fixed decoding temperature $T^*$ and taking empirical quantiles
\[
L(x) = Q^{(T^*)}_x(\alpha/2), \qquad 
U(x) = Q^{(T^*)}_x(1-\alpha/2),
\]
or (ii) by directly prompting the model to output $L(x)$ and $U(x)$.
In both cases, the construction below treats $[L(x),U(x)]$ as an initial, possibly miscalibrated interval.

\paragraph{Split-conformal calibration.}
On a held-out calibration set $\mathcal{C} = \{(x_i, y_i)\}_{i=1}^n$, we compute nonconformity scores that measure how far each true label $y_i$ lies outside its predicted interval:
\[
s_i = \max\{\, L(x_i) - y_i,\; y_i - U(x_i)\,\}.
\]
Let $s_{(1)} \le s_{(2)} \le \dots \le s_{(n)}$ denote the sorted scores, and let
\[
k = \lceil (1-\alpha)(n+1) \rceil, \qquad q_{1-\alpha} = s_{(k)}.
\]
We then form the conformalized prediction interval for any new input $x$ as
\[
\text{CI}^{\text{conf}}(x) = [\, L(x) - q_{1-\alpha},\; U(x) + q_{1-\alpha} \,].
\]

\paragraph{Coverage guarantee.}
Under exchangeability of the calibration and test samples, the conformalized interval satisfies
\[
\Pr\!\big\{\, y \in \text{CI}^{\text{conf}}(x) \,\big\} \ge 1 - \alpha,
\]
providing finite-sample, distribution-free marginal coverage.~\cite{romano2019conformalized} 
This guarantee holds regardless of how the base intervals were constructed (via multi-sampling, parametric prediction, or direct outputs), 
and is unaffected by model miscalibration, temperature choice, or distributional shape. 
In practice, well-calibrated base intervals yield smaller conformal corrections $q_{1-\alpha}$ and thus sharper final intervals, 
while the conformal procedure itself ensures correctness even when the base intervals are mis-specified.
\subsection{Direct model elicitation via Log Probabilities}

\paragraph{Log-probability extraction.}
For each question, we query the model for first-token top-$K$ log probabilities under an integer-only response format. These probabilities are normalized into a discrete distribution over integers $\widehat{p}(v)$, which represents the model’s direct belief over possible magnitudes of the answer.

\paragraph{Confidence set construction.}
Given a coverage target $c\in(0,1)$, we form the smallest-mass \emph{discrete confidence set} $S_c\subset\mathbb{Z}$ by sorting integers by $\widehat{p}(v)$ (descending, tie-breaking by value) and including values until the cumulative probability exceeds $c$. For scoring the confidence set we use its interval hull $[L_c, U_c]$ with $L_c=\min S_c$ and $U_c=\max S_c$.

\paragraph{Temperature calibration.}
In generation mode, coverage is computed via membership $1[y\in S_c]$. In post-process analysis, we apply temperature scaling $\widehat{p}_T(v)\propto \widehat{p}(v)^{1/T}$, rebuild $S_c(T)$, and evaluate coverage via $1[y\in[L_c(T),U_c(T)]]$ together with the Winkler interval score. A single temperature $T^\star$ is selected on the training split to minimize the average interval score across targets, and this calibrated $T^\star$ is then applied to the test split.

\section{Literature Review}
While ~\cite{kalyan-etal-2021-much} examined reasoning processes in Fermi problem solving, their work primarily focused on the reasoning in arriving at an answer rather than the evaluation of uncertainty quantification through confidence intervals. On the other hand, the calibration of LLMs has emerged as a critical area of study, with recent work examining various aspects of model confidence and reliability ~\cite{kapoor-etal-2024-calibration,chen2023closelookcalibrationpretrained,geng-etal-2024-survey,10.1145/3711896.3736569}. Our work extends this line of research by explicitly addressing the calibration of uncertainty bounds in approximate estimation tasks. 

Recent efforts have established principled quantitative metrics for automatic LLM evaluation~\cite{epstein2024mmmtifchallengingmultimodalmultiturn}, providing rigorous frameworks for assessing model capabilities. Concurrently, the mathematical reasoning abilities of LLMs have been extensively benchmarked through various datasets, revealing both strengths and limitations in numerical reasoning tasks~\cite{liu-etal-2024-mathbench,seßler2024benchmarkinglargelanguagemodels,xu2025ugmathbench,hong2025benchmarkingllmsmathematicalreasoning}. However, these benchmarks primarily focus on exact computation rather than order-of-magnitude estimation with calibrated uncertainty. 



\section{Dataset Construction}

We construct our benchmark from publicly available Science Olympiad Fermi questions, sourced from the repository at
\url{https://github.com/landy8697/open-scioly-fermi/tree/master}.
Each item consists of a natural language question and an approximate ``order of magnitude'' answer. 
Answers are provided as integers in base-10 exponents: for example, an answer of $-1$ corresponds to $10^{-1} = 0.1$,
an answer of $3$ corresponds to $10^{3} = 1000$, and so forth. 
Thus, each ground-truth label represents the power of ten that best approximates the true solution. 
The dataset is split into a training set with 500 questions and a test set with 500 questions, where we filter to include questions where the true answer is between $10^{-100}$ and $10^{100}$. The questions spans a diverse set of domains, 
including physics, biology, chemistry, earth sciences, and everyday reasoning. The questions also include a ground truth answer.
We note that the Fermi style is well-suited for probing calibration, since exact values are less important 
than correctly identifying the correct order of magnitude. Table~\ref{tab:fermi_examples} shows several representative examples from the dataset.

\begin{table*}[th]
\centering
\caption{Example Fermi questions and their ground-truth order-of-magnitude answers (as base-10 exponents).}
\label{tab:fermi_examples}
\begin{tabular}{p{10cm}c}
\toprule
\textbf{Question} & \textbf{Answer (Exponent)} \\
\midrule
How many pennies would it take to cover the state of Pennsylvania? & 13 \\
Suppose we had enough soup cans (2 inch radius, 6 inches tall) to fill up an Olympic sized pool. How many feet of wrapping would we need to label all of them? & 6 \\
If the weight of all humans on Earth is approximately one trillion pounds, how many pounds of ants are there on Earth? & 11 \\
Suppose a tennis player serves a ball at 120 mph. How much force, in newtons, is exerted on the ball? & 3 \\
\bottomrule
\end{tabular}
\end{table*}

\section{Models}
We use the following frontier models for evaluation: GPT-4o-mini from OpenAI, Claude-3.5-Haiku from Anthropic, and Grok-3-mini from xAI. We chose these, as they represent small cost-effective frontier models. The models are queried with the default temperature setting unless noted otherwise. We also run experiments on a set on high-performing open-source models: Qwen2.5-32B-Instruct, OpenHermes-2.5-Mistral07B, and TinyLlama-1.1B-Chat-v1.0. 

Details about the system prompt given to the models are in Appendix~\ref{llm_prompt}. 

\section{Results}
\paragraph{Metrics}
\label{Metrics}
Assume we have \(N\) questions with numeric ground truth answers \(y_i\) and model-produced intervals \([L_i, U_i]\) (\(U_i \geq L_i\)). 
We define our main scoring metric as the $\WS$:
\begin{align*}
\WS &= \frac{1}{N}\sum_{i=1}^{N}\Big[
\left(U_i - L_i\right) \\
&\qquad {}+ \frac{2}{\alpha}\,\bigl|\,y_i - \operatorname{proj}_{[L_i,U_i]}(y_i)\,\bigr|
\Big]
\end{align*}

\[
\operatorname{proj}_{[L_i,U_i]}(y)
=
\min\!\bigl\{\max\{y,\,L_i\},\,U_i\bigr\}.
\]
where the first term penalizes excessively wide intervals and the second term measures coverage (whether the interval contains the true answer).
The trade-off parameter \(\alpha > 0\) balances accuracy and sharpness. Lower scores are better.
We define the coverage as the fraction of the model answers that are within the specified confidence level.
\paragraph{Calibration}
To evaluate whether models' stated confidence levels align with empirical coverage, 
we prompt each system to produce confidence intervals at target levels 
($90\%$, $95\%$, $99\%$, $99.8\%$). 
Figure~\ref{fig:calibration_curve} plots the nominal coverage on the $x$-axis 
against the observed coverage on the $y$-axis for the base setup along with the confidence intervals where conformal prediction is applied. 
A perfectly calibrated model would lie on the dashed identity line, 
labeled ``Perfect calibration.'' 
Deviations below this line indicate overconfidence (intervals too narrow), 
while deviations above the line indicate underconfidence (intervals too wide). 
We find that all evaluated models display substantial miscalibration: 
the observed coverage falls short, 
suggesting that models systematically underestimate uncertainty. It's also interesting to note  that the observed coverage remains broadly constant with increased nominal coverage, suggesting that LLMs find it challenging to distinguish between rare events of different magnitudes.
When conformal calibration is applied according to the theory in Section~\ref{conformal_coverage}, all the models achieve near perfect calibration, showing the efficacy of the proposed conformal calibration method. 

In Appendix~\ref{additional_results} we show that several popular open source models also exhibit overconfidence in their confidence intervals.

\begin{figure*}[t]
\centering
\begin{minipage}[t]{0.47\linewidth}
    \vspace{0pt} 
    \centering
    \includegraphics[width=\linewidth]{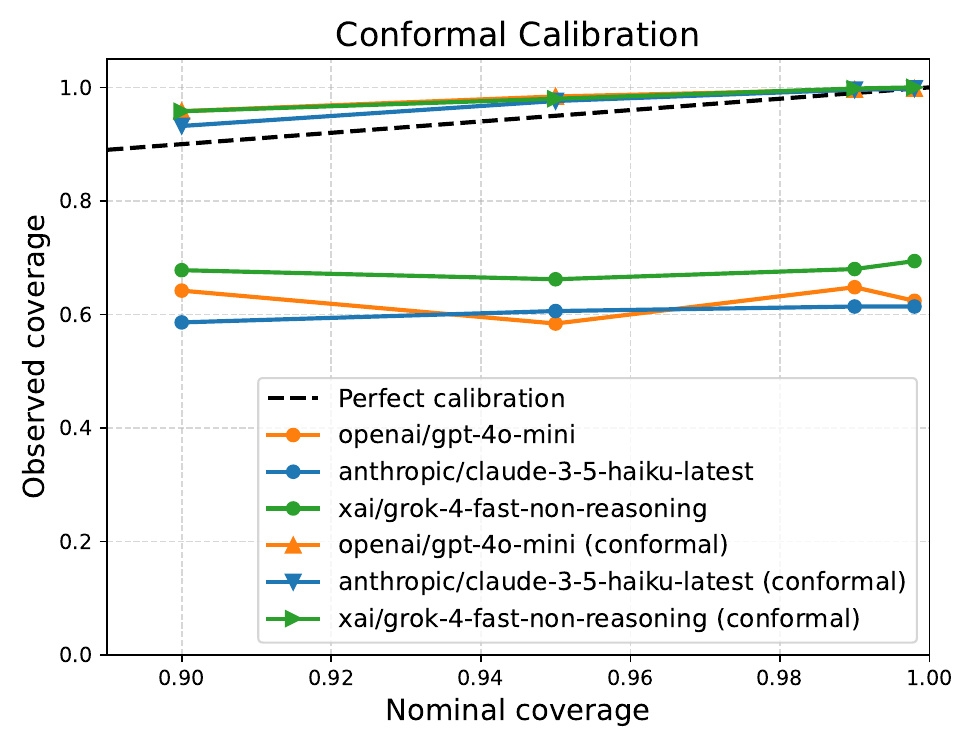}
    \caption{Calibration curves for representative models. 
    The dashed line indicates perfect calibration ($y=x$). 
    Observed coverage is consistently below nominal coverage, 
    revealing systematic overconfidence in the intervals produced by current LLMs.}
    \label{fig:calibration_curve}
\end{minipage}\hfill
\begin{minipage}[t]{0.5\linewidth}
    \vspace{0pt} 
    \centering
    \setlength{\tabcolsep}{5pt}
    \renewcommand{\arraystretch}{1.1}
    \small
    \begin{tabular}{llrrr}
    \toprule
    p ($\alpha$) & method & 
    \makecell[c]{anthropic\\claude-3-5\\haiku\\latest} & 
    \makecell[c]{openai\\gpt-4o\\mini} & 
    \makecell[c]{xai\\grok-4\\fast-non\\reasoning} \\
    \midrule
    \multirow{3}{*}{p=90 (\(\alpha=0.10\))} 
     & base & 30.96 & 27.92 & 28.93 \\
     & conformal & 21.57 & 20.78 & 23.32 \\
     & logprob & - & 33.32 & - \\
    \midrule[0.8pt]
    \multirow{3}{*}{p=95 (\(\alpha=0.05\))} 
     & base & 63.73 & 55.58 & 59.25 \\
     & conformal & 43.62 & 42.68 & 42.57 \\
     & logprob & - & 52.70 & - \\
    \midrule[0.8pt]
    \multirow{3}{*}{p=99 (\(\alpha=0.01\))} 
     & base & 321.91 & 237.57 & 318.56 \\
     & conformal & 122.24 & 108.98 & 125.31 \\
     & logprob & - & 191.03 & - \\
    \bottomrule
    \end{tabular}
    \caption{Base vs.\ Conformal vs.\ Logprob scores on test data (lower is better). 
    Conformal calibration consistently improves over the base across models and $p$. 
    The logprob heuristic helps at stricter targets ($p\ge 0.95$) but can underperform at $p=0.90$.}
    \label{tab:base-vs-conformal-merged}
\end{minipage}
\end{figure*}

\subsection{Scoring all of the Models}
In table \ref{tab:base-vs-conformal-merged}, we evaluate various models and methods. Across all targets, conformal uniformly lowers the score relative to the base method for every model. Averaged over models, the reductions are \(\approx\!25.2\%\) at \(p{=}0.90\), \(\approx\!27.8\%\) at \(p{=}0.95\), and \(\approx\!59.4\%\) at \(p{=}0.99\) (absolute drops of \(7.38\), \(16.56\), and \(173.84\) respectively). Improvements grow with stricter coverage, suggesting that raw intervals are increasingly miscalibrated at high \(p\) and benefit most from calibration. Model ranking is stable: \texttt{gpt-4o-mini} typically yields the lowest scores; the only exception is \(p{=}0.95\) under conformal where \texttt{grok-4 fast-nonreasoning} is marginally better (\(42.57\) vs.\ \(42.68\), a \(0.11\) gap). The logprob analysis (available for \texttt{gpt-4o-mini} only as the other APIs do not expose a sufficient (top 20) log probabilities) helps at higher targets (\(55.58\!\to\!52.70\) at \(p{=}0.95\); \(237.57\!\to\!191.03\) at \(p{=}0.99\)) but degrades at \(p{=}0.90\) (\(27.92\!\to\!33.32\)). Overall, conformal calibration is the dominant post-hoc adjustment in this setting, with the largest gains in the high-coverage regime.

\subsection{Evaluating the Log-Probability Analysis}
\noindent\textbf{Discussion.} Table~\ref{tab:logprob-vs-base-coverage-gpt4omini} shows that the baseline under-covers at all targets (e.g., $0.63$ at $p{=}0.90$), whereas the log-probability heuristic brings observed coverage noticeably closer to the nominal levels across all three settings. The method is inexpensive and fast: it requires no held-out calibration split, reuses token log-likelihoods produced during generation, and applies as a single-pass, vectorizable post-process. By contrast, conformal calibration needs a calibration set and per-$\alpha$ quantile estimation, both of which entail significant computational overhead. Drawbacks of this technique include the absence of finite-sample guarantees, as well as sensitivity to decoding temperature and prompt length. Given its negligible overhead, it is a practical choice when calibration data are unavailable and a solid baseline to compare against heavier calibration methods.

\begin{table}[t]
\centering
\small
\begin{tabular}{cccc}
\toprule
Nominal $p$ & Baseline obs. cov. & Logprob obs. cov. \\
\midrule
0.90 & 0.63 & 0.85  \\
0.95 & 0.59 & 0.89  \\
0.99 & 0.64 & 0.90  \\
\bottomrule
\end{tabular}
\caption{Observed coverage for \texttt{openai/gpt-4o-mini} evaluating the baseline's observed coverage versus the log probability technique's observed coverage.}
\label{tab:logprob-vs-base-coverage-gpt4omini}
\end{table}

\section{Theory of LLM Perception}

    For a distribution $F  \in \Delta((-\infty, \infty))$, we denote the restriction of $F$ to the region of mass from $i$ to $i+\beta$ as $F_{i, i+\beta}$ for any $i, \beta \ge 0$, $i+\beta \le 1$. In the case when $F$ is continuous, we will have that $F_{i, i+\beta}$ is simply the distribution $F$ conditioned on that the value is in $\left[F^{-1}(i), F^{-1}(i+\beta)\right]$.

    We posit the following theory, which can be considered mathematically as an assumption.
    \begin{theory}
    \label{theory}
        (LLM Perception Tunnel)
        Let a probability distribution of the ground truth inferred by LLM from the possible evidence and computation is $F \in \Delta((-\infty, \infty))$.
        
        For each LLM, for each question, there exists a constant mass $\beta \in (0,1]$ such that, in each query, the LLM only perceive $F_{I, I+\beta}$ where $I \sim \text{Unif}([0, 1-\beta])$, and use $F_{I, I+\beta}$ in the inference as if it is $F$.
    \end{theory}
    This theory suggests that LLM has a narrower perception on what the inferred probability distribution should be, resulting in the use of tail-trimmed distribution in $F_{i,i+\beta}$ form instead.

    Therefore, the $\frac{\alpha}{2}$ and $1-\frac{\alpha}{2}$ quantiles given from LLM is essentially the $\frac{\alpha}{2}$ and $1-\frac{\alpha}{2}$ quantiles of $F_{I, I+\beta}$ instead of those of $F$.

Note that our theory also supports the empirical findings that the observed coverage tends to stagnate at a certain value when the prompt asks for a coverage of almost $1$. The stagnation value can be a good estimate of the perception size $\beta$. 

However, we will see that under a mild regularity condition on the distribution $F$, we do not need the knowledge of $\beta$ to create a consistent estimate of the quantiles.

\begin{theorem}
\label{theorem}
    (Consistent Estimator for Tails)
    Consider a continuous distribution $F \in \Delta((-\infty, \infty))$ with a support being a closed interval.
    
    Assume that the theory~\ref{theory} is true, meaning that, in each query $i \in \mathbb{N}$, the perceived distribution is $F_{I_i, I_i+\beta}$, where $I_i \overset{\text{iid}}{\sim} \text{Unif}([0,1-\beta])$. 
    
    We further assume that the LLM has a perfect computation with respect to $F_{I_i, I_i+\beta}$, and output the lower bound $L_i = F_{I_i, I_i+\beta}^{-1}\left(\frac{\alpha}{2}\right)$, and the upperbound $U_i = F_{I_i, I_i+\beta}^{-1}\left(1-\frac{\alpha}{2}\right)$.

    By defining
    \begin{align*}
        \hat{L}_n = \text{empirical quantile}\left([L_i]_{i=1}^n, \frac{\alpha}{2}\right),
    \end{align*}
    and
    \begin{align*}
        \hat{U}_n = \text{empirical quantile}\left([U_i]_{i=1}^n, 1-\frac{\alpha}{2}\right),
    \end{align*}
    we will have that
    \begin{align*}
        \left(\hat{L}_n, \hat{U}_n\right) 
        \overset{p}{\to} \left(F^{-1}\left(\frac{\alpha}{2}\right), F^{-1}\left(1-\frac{\alpha}{2}\right)\right).
    \end{align*}
\end{theorem}

The theorem~\ref{theorem} suggests that we can find a consistent estimate of $\frac{\alpha}{2}$ and $1-\frac{\alpha}{2}$ quantiles of the underlying distribution $F$ from the empirical quantiles of the observed $\frac{\alpha}{2}$ and $1-\frac{\alpha}{2}$ quantiles output by the LLM. Moreover, the consistency is independent from the value of $\beta$, meaning that the difficulty in perceiving a more complete distribution does not change asymptotic convergence though it can affect the rate.

Note that when the perception of the LLM is perfect in the sense that $\beta =1$, $L_i$ is a degenerate distribution at the value of $F^{-1}\left(\frac{\alpha}{2}\right)$, so our method also works. In another limit case when $\beta$ is approaching $0$, we can think of $L_i$ as a sample independently drawn from $F$, so our method also works in such scenario.

Under the theory of LLM perception tunnel, we do also have the rate of convergence in addition to the consistency for the estimators, as shown in Appendix~\ref{app_theory}.

\section{Conclusion}
We introduced FermiEval, a benchmark for evaluating how well large language models calibrate confidence intervals on numerical estimation tasks. Across modern models, we find systematic overconfidence, with nominal 99\% intervals covering the truth only about 65\% of the time. We propose a conformal calibration method that restores accurate coverage and halves the Winkler score, and establish a theoretical framework as an hypothesis on why LLMs show overconfidence.

\setcounter{secnumdepth}{1}
\bibliography{aaai2026}
\appendix
\section{Proofs and additional theory}
\label{app_theory}
\subsection{Proof of Theorem~\ref{theorem}}
\begin{proof}
    Note that $F(L_i) = F\left(F_{I_i, I_i+\beta}^{-1}\left(\frac{\alpha}{2}\right)\right) = I_i + \frac{\alpha \beta}{2}$. This makes $F(L_i) \overset{\text{iid}}{\sim} \text{Unif}\left(\left[\frac{\alpha \beta}{2}, 1-\beta+\frac{\alpha \beta}{2}\right]\right)$, whose $\frac{\alpha}{2}$ quantile is
    \begin{align*}
        \left(1-\frac{\alpha}{2}\right)\frac{\alpha \beta}{2} + \frac{\alpha}{2} \left(1-\beta+\frac{\alpha \beta}{2}\right) = \frac{\alpha}{2}.
    \end{align*}
    Therefore, from strict monotonicity of $F$ in the interior of the closed interval support of $F$, we will have that $F\left(\hat{L}_n\right) \overset{p}{\to} \frac{\alpha}{2}$. From continuous mapping theorem, $\hat{L}_n \overset{p}{\to} F^{-1}\left(\frac{\alpha}{2}\right)$. The proof can also be applied to the convergence of $\hat{U}_n$.
\end{proof}
\subsection{Convergence Rate of Estimator}
\begin{proposition}
\label{proposition}
    (Rate of Convergence of the Estimators)
    Under the same condition in theorem~\ref{theorem}, we will have that
    \begin{align*}
        \sqrt{n}\left(\begin{bmatrix}
            \hat{L}_n\\
            \hat{U}_n
        \end{bmatrix} - 
        \begin{bmatrix}
            F^{-1}\left(\frac{\alpha}{2}\right)\\
            F^{-1}\left(1-\frac{\alpha}{2}\right)
        \end{bmatrix}
        \right)
        \overset{d}{\to}
    \end{align*}
    \begin{equation*}
    \resizebox{\columnwidth}{!}{$
    N\!\left(0,(1-\beta)
    \begin{bmatrix}
    \frac{\frac{\alpha}{2}\left(1-\frac{\alpha}{2}\right)}{f\!\left(F^{-1}\!\left(\frac{\alpha}{2}\right)\right)^{2}} &
    \frac{\frac{\alpha}{2}\left(1-\frac{\alpha}{2}\right)}{f\!\left(F^{-1}\!\left(\frac{\alpha}{2}\right)\right)f\!\left(F^{-1}\!\left(1-\frac{\alpha}{2}\right)\right)} \\
    \frac{\frac{\alpha}{2}\left(1-\frac{\alpha}{2}\right)}{f\!\left(F^{-1}\!\left(\frac{\alpha}{2}\right)\right)f\!\left(F^{-1}\!\left(1-\frac{\alpha}{2}\right)\right)} &
    \frac{\frac{\alpha}{2}\left(1-\frac{\alpha}{2}\right)}{f\!\left(F^{-1}\!\left(1-\frac{\alpha}{2}\right)\right)^{2}}
    \end{bmatrix}\right)
    $}
    \end{equation*} 
\end{proposition}

\begin{proof}
    Since $F(L_i) \overset{\text{iid}}{\sim} \text{Unif}\left(\left[
        \frac{\alpha \beta}{2}, 1-\beta+\frac{\alpha \beta}{2}
    \right]\right)$, we will have that 
    \begin{align*}
        \sqrt{n}\left[F\left(\hat{L}_n\right) - \frac{\alpha}{2}\right]
        \overset{d}{\to}
        \mathcal{N}\left(0,(1-\beta)\frac{\alpha}{2}\left(1-\frac{\alpha}{2}\right)\right).
    \end{align*}
    We have the $F(U_i) = F(L_i) + \frac{1-\alpha}{\beta}$, so
    \[
\begin{aligned}
&\sqrt{n}\!\left(
\begin{bmatrix}
F(\hat{L}_n)\\[2pt]
F(\hat{U}_n)
\end{bmatrix}
-
\begin{bmatrix}
\frac{\alpha}{2}\\[2pt]
1-\frac{\alpha}{2}
\end{bmatrix}
\right)
\overset{d}{\to}\\
&\mathcal{N}\!\Bigl(0,\,(1-\beta)\tfrac{\alpha}{2}\!\left(1-\tfrac{\alpha}{2}\right)
\begin{bmatrix}
1 & 1\\
1 & 1
\end{bmatrix}\Bigr).
\end{aligned}
\]
    Apply delta method and the existence of pdf (probability distribution function) $f$ being positive in the interior of the support interval yields the result.
\end{proof}
This proposition suggests that the perception size $\beta$ will affect the constant in the asymptotic normality of the estimator, but will not change the rate of convergence unless $\beta$ is exactly $1$.
\section{Scoring Rule}

    Although we frame it as a confidence interval, we are actually looking at credible interval using a subjective belief of the model. Assume that the model, with the evidence it can gather and recall, formulate a subjective belief that the answer is a random variable $X \sim F$. Note that we want the distribution $F \in \Delta((0, \infty))$ to be ``consistent" in the sense that different recall or instantiation of the LLM will still give the same $F$. Otherwise, if $F$ itself is stochastic, we wish that the LLM will use the average (over the randomness of $F$) $F$ as a subjective probability distribution instead. In other words, we want the subjective belief to include the belief in other instantiation.\footnote{Why do we need this consistency? Without this consistency, LLM can defend that, in such instantiation, its subjective distribution $F$ is very narrow.} 

    \subsection{$\left[\frac{\alpha}{2}, 1-\frac{\alpha}{2}\right]$ Interval}

        We will use Winkler interval score, which is a proper scoring rule.~\cite{winkler1972decision} If the random variable $X$ is realized to be $x$, we calculate the loss to be
        \begin{align*}
            l(L, U, x) &:=
            \left[\log U - \log L\right] \\
            &+ \frac{2}{\alpha} \left\vert \log x - \text{Proj}_{[\log L, \log U]}(\log x)\right\vert.
        \end{align*}
        A class of Winkler interval score is broader and proprierty has been proven for general distribution in~\cite{gneiting2007strictly}. Here, we provide the proof specific to our scoring rule for intuition.
        \begin{proposition}
        \label{propsotion:winkler}
            For a distribution $F \in \Delta((0, \infty))$, there exists a unique $(L^*, U^*) \in (0, \infty)^2$ with $L^* \le U^*$ solving
            \begin{align*}
                \inf_{L,U > 0: L \le U} \mathbb{E}_{X \sim F}\left[ l(L, U, X)\right].
            \end{align*}
            Moreover, 
            \begin{align*}
            (L^*,U^*) &=
            \left( \inf\{\,y>0:\, F(y)\ge \alpha/2\,\}, \right.\\
            &\qquad \left. \inf\{\,y>0:\, F(y)\ge 1-\alpha/2\,\} \right)
            \end{align*}
        \end{proposition}
        \begin{proof}
            Note that $\mathbb{E}_{X \sim F}\left[ l(L, U, X)\right]$ is convex in $(\log L, \log U)$. The subderivative 
            \begin{align*}
                &\partial_{\log L}\mathbb{E}_{X \sim F}\left[ l(L, U, X)\right]\\
                &= \{-1\} +\partial_{\log L}\left(
                    \frac{2}{\alpha}
                    \mathbb{E}_{X \sim F}\left[
                      (\log L - \log X) \mathbf{1}_{X \le L}
                   \right]
                \right)\\
                &=
                \left[
                    \frac{2}{\alpha}F((0,L))-1,
                    \frac{2}{\alpha}F((0,L])-1
                \right].
            \end{align*}
            Thus, $L^* = \inf\left(\left\{y>0: F(y)\ge\frac{\alpha}{2} \right\}\right)$, and we can find $U^*$ in a similar manner.
        \end{proof}

    From the proposition~\ref{propsotion:winkler}, we will have that
        \begin{align*}
            F([L^*, U^*]) \ge 1-\alpha,
        \end{align*}
        where the equality case is ensured when the distribution $F$ is continuous.

    \paragraph{Theory Visualization}
    \begin{figure}[h]
    \centering
    \includegraphics[width=0.85\linewidth]{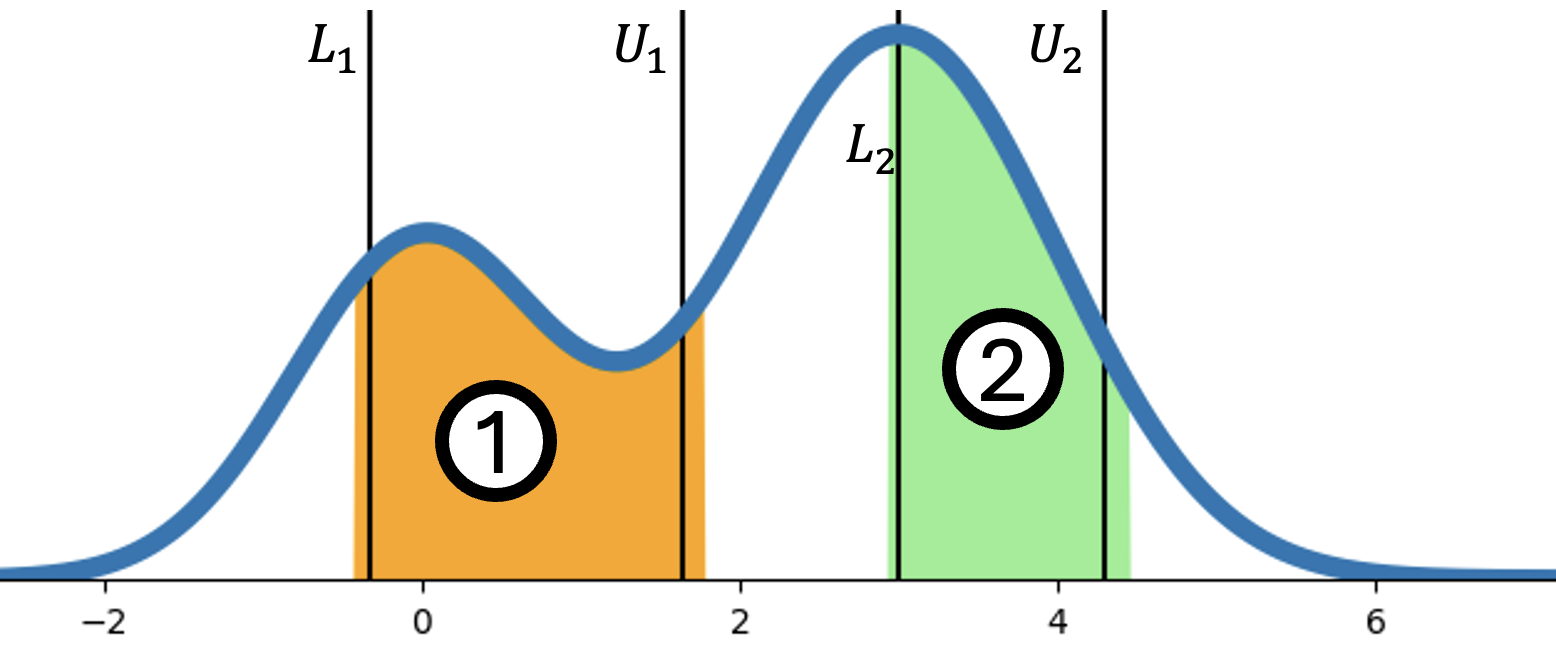}
    \caption{LLM Vision Tunnel: The inferred distribution has its pdf shown the bold curve. However, in each query, LLM only perceives a section of the distribution. For example, in the first query, LLM may only perceive the orange distribution, while it can perceive the green distribution in the second query. The answer for the lower bound and upper bound are then different.}
    \label{fig:fermi_hist_zoom3}
    \end{figure}

    \begin{figure}[h]
    \centering
    \includegraphics[width=0.85\linewidth]{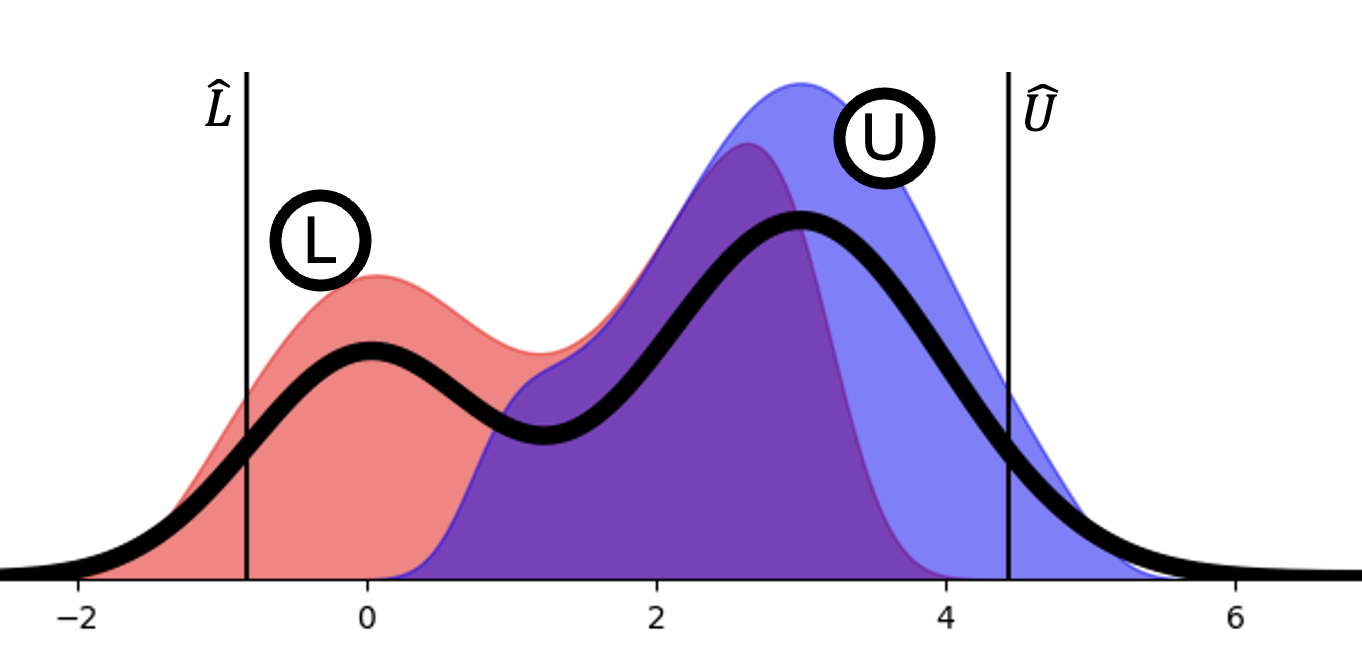}
    \caption{The distribution of lower bound $L$ as a random variable where the randomness stemmed from the random vision tunnel of LLM is shown as a red distribution, while that of the random upper bound $U$ is shown as a blue distribution. We will see that there is an overlap between the two distributions. The original distribution $F$ is displayed by its pdf as a bold curve. By denoting $\hat{L}$ to be the $\frac{\alpha}{2}$ quantile of the distribution of $L$ and $\hat{U}$ to be the $1-\frac{\alpha}{2}$ quantile of the distribution of $U$, we will have that both of them will also serve as $\frac{\alpha}{2}$ and $1-\frac{\alpha}{2}$ quantiles of the original distribution $F$. Note that this result is independent of the perception size $\beta$.}
    \label{fig:fermi_hist_zoom2}
    \end{figure}
    Figure~\ref{fig:fermi_hist_zoom3} and Figure~\ref{fig:fermi_hist_zoom2} visualizes the theory. 
\section{Additional Results}
\label{additional_results}
\paragraph{Calibration results for open source models}
Figure~\ref{fig:local_models} shows the calibration curves for several open-source models. Here, we note that in the same way as for the frontier models, the local models exhibit significant overconfidence. 

\begin{figure}
    \centering
    \includegraphics[width=0.9\linewidth]{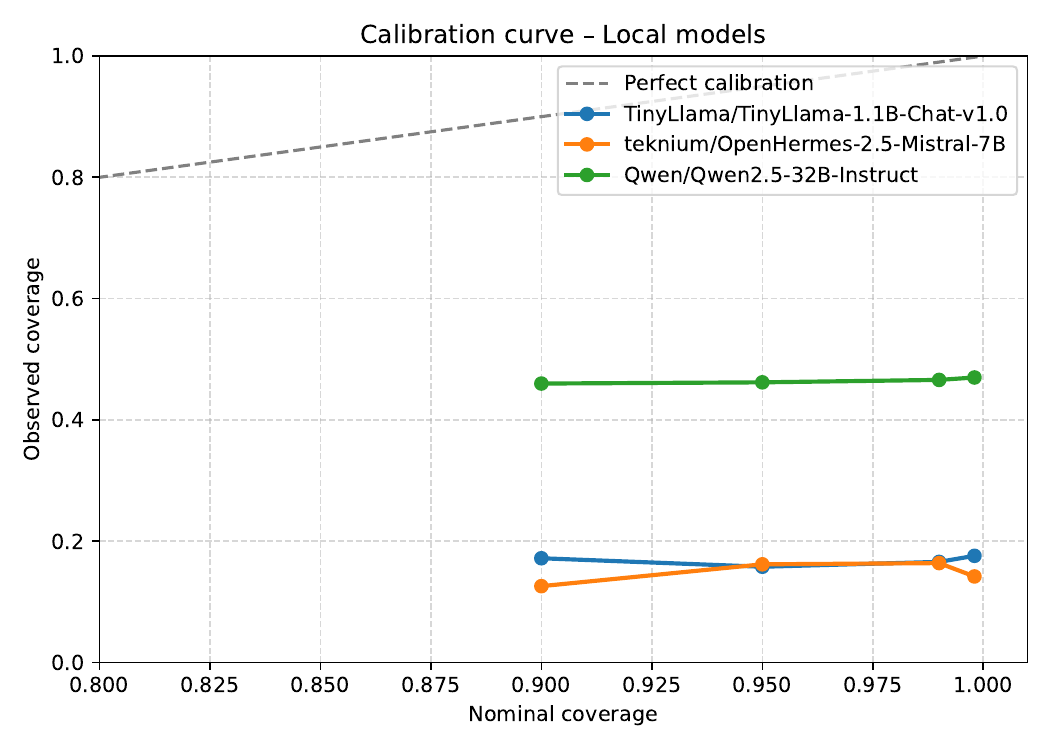}
    \caption{Calibration curves for open-source models. 
The dashed line indicates perfect calibration ($y=x$). 
Observed coverage is consistently below nominal coverage, 
revealing systematic overconfidence in the intervals produced by current LLMs.}
    \label{fig:local_models}
\end{figure}



\section{Additional Dataset Details}
\begin{figure}[h]
\centering
\includegraphics[width=0.85\linewidth]{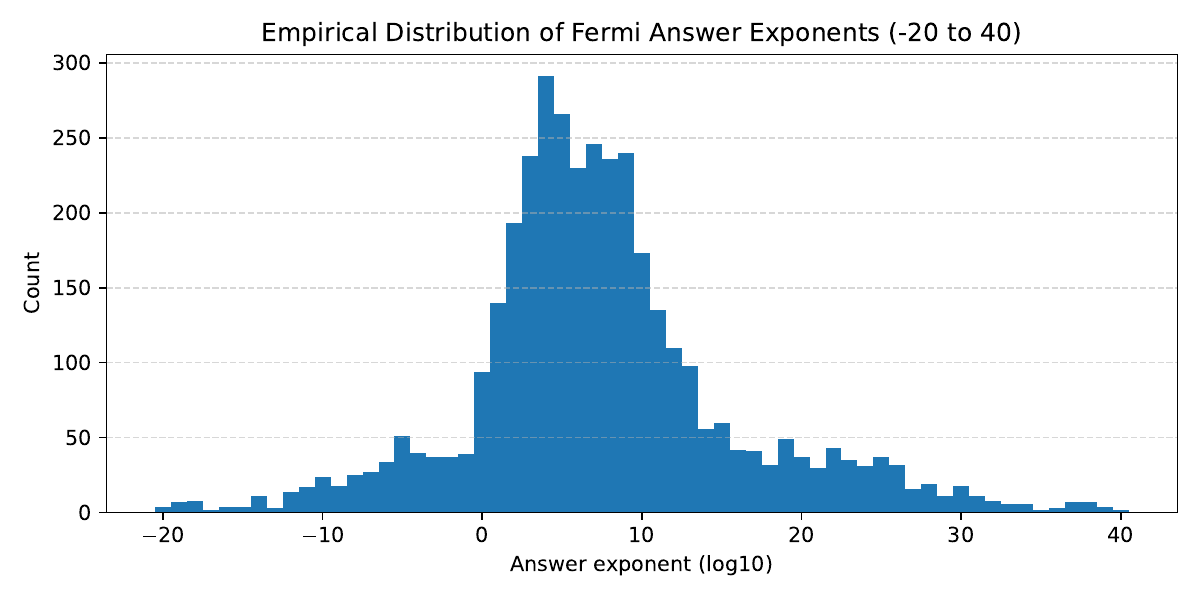}
\caption{Empirical distribution of ground-truth order-of-magnitude answers (base-10 exponents) \emph{restricted to} the interval $[-20, 40]$. 
This zoomed view omits extreme tails for readability; nevertheless, most questions in the corpus lie within this range.}
\label{fig:fermi_hist_zoom}
\end{figure}

Figure~\ref{fig:fermi_hist_zoom} shows the empirical distribution of answer exponents \emph{restricted to} the range $[-20, 40]$ for a superset of the FermiEval train and test set; 
while this plot excludes the long tails for readability, the vast majority of questions fall within this range.

\section{LLM Prompt}
\label{llm_prompt}
\paragraph{Model Prompt}

For the experiments in Section~\ref{Metrics}, we use the following prompt:

\begin{lstlisting}
System Instruction:
    "You are a careful quantitative estimator for Fermi-style questions.
    Provide brief reasoning if useful, but you MUST END with a single JSON\n object of the form {{"L": <int>, "U": <int>}} where L and U are INTEGERS (can be negative), representing a {pct} percentage confidence interval [10^L, 10^U]. If >={pct} percentage of the mass is on exactly 10^x,
    use L=U=x. \nDo not include units in the JSON. Ensure L <= U. Keep the response less than 256 tokens."

User prompt: (f"Question: {question}\n\n"
            f"Give a {pct:.2f}% confidence interval in the form\n [10^L, 10^U] with INTEGER L and U.\n"
            f"Finish with JSON ONLY: {"L": <int>, "U": <int>}.")
\end{lstlisting}

\section{Temperature Calibration via Multi-Sampling}
\label{multi-sampling}
This section presents an additional way to improve model confidence interval estimation, that was ultimately not pursued for empirical experiments due to the high sampling cost associated with the method. 

To improve the reliability of model confidence intervals, we tune the decoding temperature $T$ so that empirical coverage on a validation set matches the target level while keeping intervals as sharp as possible. The procedure consists of two stages: (i) fitting $T$ on a training set, and (ii) applying the calibrated temperature to produce confidence intervals on the test set.

\paragraph{Step 1: Fitting the decoding temperature.}
We seek a temperature $T^*$ that achieves the desired empirical coverage (e.g., 95\%) with the tightest possible intervals. We evaluate a grid of candidate temperatures, e.g. $T \in \{0.2, 0.4, 0.6, 0.8, 1.0, 1.2, 1.5\}$. For each $T$ and each validation item $i$ with ground truth $y_i^{\text{true}}$, we draw $M$ independent samples $y_{i,1:M}^{(T)}$ using decoding temperature $T$, and compute the empirical quantiles
\[
L_i^{(T)} = Q_i^{(T)}(\alpha/2), \qquad U_i^{(T)} = Q_i^{(T)}(1-\alpha/2),
\]
which define a $(1-\alpha)$ confidence interval for that item.

For each temperature, we compute (i) the empirical coverage
\[
\widehat{\text{coverage}}(T) = \frac{1}{N} \sum_i \mathbf{1}\!\big[L_i^{(T)} \le y_i^{\text{true}} \le U_i^{(T)}\big],
\]
and (ii) the mean Winkler interval score
\begin{align*}
\overline{S}(T)
&= \frac{1}{N}\sum_i \left[
  (U_i^{(T)} - L_i^{(T)}) \right.\\
&\qquad {}+ \tfrac{2}{\alpha}(L_i^{(T)} - y_i^{\text{true}})\mathbf{1}_{y_i^{\text{true}} < L_i^{(T)}}\\
&\qquad \left. {}+ \tfrac{2}{\alpha}(y_i^{\text{true}} - U_i^{(T)})\mathbf{1}_{y_i^{\text{true}} > U_i^{(T)}}
\right].
\end{align*}

We then select the temperature that minimizes
\[
\mathcal{J}(T) = \overline{S}(T) + \lambda \!\left[\max\!\big(0, (1-\alpha) - \widehat{\text{coverage}}(T)\big)\right]^2,
\]
where $\lambda$ penalizes under-coverage. The resulting $T^*$ is used as the calibrated decoding temperature.

\paragraph{Step 2: Generating and evaluating confidence intervals.}
We fix all decoding parameters at the tuned values and apply the fitted temperature $T^*$ to unseen test data. For each test item $j$, we sample $M$ outputs $y_{j,1:M}^{(T^*)}$ and compute
\[
L_j = Q_j^{(T^*)}(\alpha/2), \qquad U_j = Q_j^{(T^*)}(1-\alpha/2),
\]
yielding the confidence interval $\text{CI}_j = [L_j, U_j]$. We then compute empirical coverage and interval scores on the test set to verify that the calibration generalizes beyond the validation data.

\section{Tunnel Vision Effect}

During inference, a language model $p_{\theta}$ generates a \emph{thinking trace} $R$ under a decoding policy $q(R \mid X)$.
The resulting predictive distribution over answers is
\begin{equation}
p_q(Y \mid X)
   = \mathbb{E}_{R \sim q}\big[p_\theta(Y \mid X, R)\big].
\end{equation}

Define the \emph{Tunnel Index} for decoding method $q$ as the reduction in predictive entropy induced by conditioning on the internal reasoning trace:
\begin{equation}
\mathrm{TI}_q
   = H_\theta(Y \mid X)
     - \mathbb{E}_{R \sim q}\big[H_\theta(Y \mid X, R)\big].
\end{equation}

\textbf{Tunnel Vision Effect.}
Conditioning on a self-generated thinking trace reduces predictive uncertainty and increases the model’s apparent confidence.
The strength of this confidence inflation, measured by $\mathrm{TI}_q$, grows as the entropy of the decoding policy decreases:
\begin{equation}
\mathrm{TI}_q > 0,
\qquad
\frac{\partial \mathrm{TI}_q}{\partial H(q)} < 0.
\end{equation}

Low-entropy decoding policies (e.g., greedy or beam search) produce narrow $q(R \mid X)$ and large $\mathrm{TI}_q$.
High-entropy or multi-sample decoding spreads probability mass over diverse reasoning traces, reducing $\mathrm{TI}_q$ and restoring calibration.

\paragraph{Self-Consistency through the Tunnel Vision Hypothesis}

The self-consistency decoding method \citep{wang2023selfconsistencyimproveschainthought} can be viewed as an empirical remedy to the Tunnel Vision Effect.
In practice, the predictive distribution under a decoding policy $q$ is approximated via Monte Carlo sampling:
\begin{equation}
\hat{p}_m(Y \mid X)
   = \frac{1}{m}\sum_{i=1}^m p_\theta(Y \mid X, R_i),
   \qquad R_i \sim q(R \mid X).
\end{equation}

\begin{theorem}[Self-Consistency Reduces Tunnel Vision]
Define the empirical Tunnel Index
\[
\widehat{\mathrm{TI}}_m
   = H_\theta(Y \mid X) - H(\hat{p}_m(Y \mid X)).
\]
Then, for any decoding policy $q(R \mid X)$ and any $m \ge 1$,
\begin{equation}
\mathbb{E}[\widehat{\mathrm{TI}}_m] \ge \mathrm{TI}_q,
\lim_{m \to \infty} \mathbb{E}[\widehat{\mathrm{TI}}_m] = \mathrm{TI}_q.
\end{equation}
\end{theorem}

\begin{proof}
By Jensen’s inequality, $\mathbb{E}[H(\hat{p}_m)] \le H(p_q)$, which implies
$\mathbb{E}[\widehat{\mathrm{TI}}_m] \ge \mathrm{TI}_q$ and equality as $m\!\to\!\infty$.
\end{proof}

\paragraph{Interpretation.}
Greedy decoding ($m=1$) corresponds to a single low-entropy reasoning path, yielding a biased, overconfident estimate of $p_q(Y \mid X)$ and a large $\widehat{\mathrm{TI}}_1$.
As $m$ increases, self-consistency expands the entropy of the decoding policy and improves the Monte Carlo estimate of the latent marginal, thereby reducing the Tunnel Index and mitigating confidence inflation.

\end{document}